\documentclass[journal]{IEEEtran}
\usepackage{amsmath,amsfonts}
\usepackage{amssymb}
\usepackage{amsthm}
\usepackage{thmtools}
\usepackage{algorithmic}
\usepackage{algorithm}
\usepackage{array}
\usepackage[caption=false]{subfig}
\usepackage{textcomp}
\usepackage{stfloats}
\usepackage[colorlinks, linkcolor=blue, anchorcolor=blue, citecolor=blue, urlcolor=blue]{hyperref}
\usepackage{url}
\usepackage{verbatim}
\usepackage{graphicx}
\usepackage{cite}

\declaretheorem[style=plain]{theorem}
\declaretheorem[style=plain]{proposition}
\declaretheorem[style=remark,qed={}]{remark}

\begin{document}

\title{\huge Planar Diffractive Neural Networks Empowered Communications:\\A Spatial Modulation Scheme}

\author{Xiaokun Teng,\IEEEmembership{}
	Yanqing Ren,\IEEEmembership{}
	Weicong Chen,\IEEEmembership{~Member,~IEEE,}
	Wankai Tang,\IEEEmembership{~Member,~IEEE,}\\
	Xiao Li,\IEEEmembership{~Member,~IEEE,} 
    and Shi Jin,\IEEEmembership{~Fellow,~IEEE}
	
	\thanks{Xiaokun Teng, Yanqing Ren, Weicong Chen, Wankai Tang, and Shi Jin are with the School of Information Science and Engineering, Southeast University, Nanjing 210096, China (e-mail: xkteng@seu.edu.cn; renyq@seu.edu.cn; cwc@seu.edu.cn; tangwk@seu.edu.cn; li\_xiao@seu.edu.cn; jinshi@seu.edu.cn).\textit{(Corresponding authors: Shi Jin; Wankai Tang.)}}
	
}

\markboth{Journal of \LaTeX\ Class Files,~Vol.~14, No.~8, August~2021}%
{Shell \MakeLowercase{\textit{et al.}}: A Sample Article Using IEEEtran.cls for IEEE Journals}

\maketitle

\begin{abstract}
Diffractive neural networks, where signal processing is embedded into wave propagation, promise light-speed and energy-efficient computation. However, existing three-dimensional structures, such as stacked intelligent metasurfaces (SIMs), face critical challenges in implementation and integration. In contrast, this work pioneers planar diffractive neural networks (PDNNs) empowered communications, a novel architecture that performs signal processing as signals propagate through artificially designed planar circuits. To demonstrate the capability of PDNN, we propose a PDNN-based space-shift-keying (PDNN-SSK) communication system with a single radio-frequency (RF) chain and a maximum power detector. In this system, PDNNs are deployed at both the transmitter and receiver to jointly execute modulation, beamforming, and detection. We conduct theoretical analyses to provide the maximization condition of correct detection probability and derive the closed-form expression of the symbol error rate (SER) for the proposed system. To approach these theoretical benchmarks, the phase shift parameters of PDNNs are optimized using a surrogate model-based training approach, which effectively navigates the high-dimensional, non-convex optimization landscape. Extensive simulations verify the theoretical analysis framework and uncover fundamental design principles for the PDNN architecture, highlighting its potential to revolutionize RF front-ends by replacing conventional digital baseband modules with this integrable RF computing platform.

\end{abstract}


\section{Introduction}
\IEEEPARstart{T}{he} forthcoming sixth-generation (6G) mobile communication networks are envisioned to deliver ubiquitous connectivity, ultra-high energy and spectral efficiency, while enabling a wide range of emerging applications, spanning from artificial intelligence to integrated sensing and communication~\cite{tariqSpeculativeStudy6G2020a, you6GWirelessCommunication2020}. The key performance indicators (KPIs) defining this new era include peak data rates in the terabit-per-second (Tbps) range, sub-millisecond latency, and an unprecedented connection density of 10\textsuperscript{6}–10\textsuperscript{8} devices per km\textsuperscript{2} \cite{itu-rRecommendationM21600Framework2023}. While these ambitious goals are expected to revolutionize how people interact with the digital world, they also impose immense pressures on the evolution of traditional wireless communication infrastructures. For decades, wireless technologies have relied on an architecture that separates the analog radio-frequency (RF) front end from the digital baseband processing unit \cite{liangLowComplexityHybridPrecoding2014}. In this paradigm, computational tasks such as modulation, beamforming, and detection are predominantly executed in the digital domain, necessitating the use of power-hungry digital signal processors and high-resolution analog-to-digital/digital-to-analog converters (ADCs/DACs). To meet the unprecedented capacity demands of 6G, higher frequency bands and extremely large-scale antenna arrays are expected to be exploited  \cite{akyildizTerahertzBandCommunication2022, luTutorialNearFieldXLMIMO2024}. However, this evolution leads to a continuous increase in hardware complexity and power consumption in the conventional digitally-centered approach, making it less efficient from the perspective of sustainable development \cite{prasadEnergyEfficiencyMassive2017}.

In response to the above limitations, various techniques aimed at reducing the hardware cost and power consumption have been explored over the years. Among them, the reconfigurable intelligent surface (RIS) has emerged as a key motivator for 6G \cite{basarReconfigurableIntelligentSurfaces2024, liuReconfigurableIntelligentSurfaces2021, huangReconfigurableIntelligentSurfaces2019}. An RIS is a planar surface composed of many low-cost elements, each capable of inducing a controllable phase shift on the incident electromagnetic (EM) wave \cite{tangWirelessCommunicationsReconfigurable2020}. The most common application of RISs involves deploying them in the wireless environment to intelligently control signal propagation. By manipulating the phase shifts of its elements, RISs can reconfigure the wireless channel to enhance desired signals or cancel interference, thereby improving the performance and reliability of existing communication links \cite{sangCoverageEnhancementDeploying2024a, chenChannelCustomizationJoint2023, hanLargeIntelligentSurfaceassisted2019, wuIntelligentReflectingSurface2019}. Moreover, the potential of RISs extends far beyond channel enhancement. Research has explored incorporating RIS into the transceiver, leading to innovative modulation architectures that reduce hardware complexity and power consumption \cite{chengReconfigurableIntelligentSurfaces2022, tangWirelessCommunicationsProgrammable2019}. One such paradigm is RIS-based reflection modulation, where the RIS directly modulates an unmodulated carrier wave by dynamically switching between different reflection patterns. Different modulation schemes, such as frequency-shift keying (FSK) \cite{zhaoProgrammableTimedomainDigitalcoding2019}, phase-shift keying (PSK) \cite{daiWirelessCommunicationsSimplified2019}, quadrature amplitude modulation (QAM) \cite{chenAccurateBroadbandManipulations2022}, and reflection pattern modulation \cite{cuiDirectTransmissionDigital2019} can be realized through this method. Multiplexed and MIMO communication were also implemented to improve the efficiency of RIS-based transmitter \cite{zhangWirelessCommunicationScheme2021, tangMIMOTransmissionReconfigurable2020}. Another meaningful direction is RIS-empowered index modulation (IM) or spatial modulation (SM) \cite{basarReconfigurableIntelligentSurfaceBased2020}. In these schemes, information is encoded in the indices of activated RIS elements \cite{liuRealizationReconfigurableIntelligent2022a}, or transmitting and receiving antennas \cite{maLargeIntelligentSurface2020, liNovelRISAidedReceive2024}. This approach inherits the benefits of traditional SM, leading to improved efficiency and a simpler hardware structure \cite{meslehSpatialModulation2008}. Despite these architectural advancements, RIS-based modulation schemes still typically rely on digital baseband processors at the receiver to perform complex signal processing tasks, such as signal detection. Therefore, while they represent a significant step toward simpler hardware, they do not completely eliminate the reliance on digital baseband.

Motivated by the developments of RIS, the emergence of diffractive deep neural network (D²NN) \cite{linAllopticalMachineLearning2018} and D²NN-inspired architectures challenges the traditional fully-digital computing paradigm. They aim to perform signal processing directly in the EM domain as the waves propagate through transmissive RISs (T-RISs), thereby realizing the concept of ``computing-by-propagation''. A D²NN comprises multiple cascaded layers of transmissive metasurfaces, where each element on a layer acts as a secondary wave source, contributing to the diffracted field on the next layer. By carefully engineering the transmission coefficients of these layers, D²NNs can be trained to execute complex functions, such as image classification and human action recognition \cite{zhouLargescaleNeuromorphicOptoelectronic2021, chenDiffractiveDeepNeural2024}, entirely through wave diffraction. In \cite{liuProgrammableDiffractiveDeep2022}, this architecture was extended to the millimeter-wave band using T-RISs, enabling mobile communication encoding–decoding and real-time multi-beam focusing. Since such EM-domain signal processing operates at the speed of light with minimal power consumption, it offers extremely high energy efficiency and low latency. Motivated by these advantages, D²NN concepts have been further adapted to wireless communications through the proposal of RIS based deep neural networks (Rb-DNNs) \cite{wangInterplayRISAI2021} and stacked intelligent metasurfaces (SIMs) \cite{anStackedIntelligentMetasurfaces2023a}. Among these, SIM has received widespread attention in recent years, and has been explored for a range of applications, including beamforming \cite{hassanEfficientBeamformingRadiation2024}, direction-of-arrival estimation \cite{anTwodimensionalDirectionofarrivalEstimation2024}, and integrated sensing and communication (ISAC) \cite{wangMultiuserISACStacked2024}, demonstrating its potential to offload computationally intensive tasks from the digital baseband. 

However, these three-dimensional (3D) wave-based signal processors face intractable practical challenges. Firstly, their performance is critically sensitive to the precise alignment of the stacked layers; sub-wavelength misalignments can cause severe distortion. Meanwhile, the near-field mutual coupling effects between closely spaced layers are difficult to model accurately with conventional scalar diffraction theories, leading to a gap between theory and practice. More importantly, the 3D nature of these structures makes them difficult to integrate with the modern planar circuits in user equipment as computation modules, which greatly limits their application scenarios. These challenges in modeling, implementation, and integration hinder the widespread adoption of SIM in practical communication systems.

Recently, the two-dimensional (2D) version of D²NNs has emerged with the development of planar diffractive neural networks (PDNN) \cite{guDirectElectromagneticInformation2024, gaoTerahertzSpoofPlasmonic2024, gaoProgrammableSurfacePlasmonic2023}. Unlike their 3D counterparts, these networks confine and guide EM waves along the surface of a planar substrate, such as a printed circuit board (PCB). The ``diffraction'' between layers is artificially engineered using transmission line couplers. The learnable parameters of these networks are realized through arrays of tunable phase shifters, strategically placed within each layer to modulate the phase of the propagating signals. The PDNN architecture enables seamless integration of this EM neural network with other RF front-end components on a single substrate, facilitating on-device deployment in user equipment. Moreover, this planar structure inherently resolves the critical issues of layer-to-layer alignment and offers superior physical consistency. The behavior of signals within these structures can be accurately described by well-established transmission line theory, reducing the modeling uncertainties associated with near-field mutual coupling. Motivated by the advantages of this emerging platform, this paper proposes a PDNN-based space-shift-keying (PDNN-SSK) communication system. Empowered by PDNNs, the proposed system can be trained to execute joint modulation, beamforming, and detection entirely in the wave domain, therefore achieving reduced system complexity. The primary contributions of this paper are summarized as follows:
\begin{itemize}
	\item{We propose a novel PDNN-SSK communication architecture that deploys a PDNN at both the transmitter and the receiver. At the transmitter, the PDNN processes the SSK signal to synthesize the desired wavefront for radiation. At the receiver, the PDNN acts as an RF demodulator that focuses the signal energy onto a specific output port corresponding to the transmitted symbol. This allows for simple detection via non-coherent power comparison, replacing the entire digital baseband processing module with a planar RF computing network.}
	\item{We conduct a rigorous theoretical analysis of the PDNN-SSK system. For the non-coherent detection scheme, we derive the conditional correct detection probability (CCDP) and prove its monotonicity and maximization condition. Meanwhile, we provide an accurate closed-form expression for the symbol error rate (SER) under interference-free conditions, along with a tight asymptotic expression at high signal-to-noise ratios (SNRs). These derivations offer crucial theoretical benchmarks for the system performance.}
	\item{We formulate the phase configuration of the PDNNs as an optimization problem aimed at maximizing the effective achievable sum-rate. To solve this high-dimensional, non-convex problem, we abstract the entire PDNN-SSK communication system into a differentiable surrogate model. This allows us to efficiently find a high-quality solution using a gradient-based optimizer.}
	\item{We conduct extensive simulations to verify the theoretical analysis framework and uncover fundamental design principles for the proposed architecture. The closed-form SER expressions are validated via Monte Carlo simulations. Meanwhile, the PDNN-SSK scheme is benchmarked against existing modulation technologies, highlighting its superior scalability and energy efficiency. The efficiency of the optimization method is underscored by comparison with conventional gradient ascent methods. Further results reveal the trade-offs among network depth, width, and artificial coupling strength, shedding light on the optimal design of PDNN structural parameters.}
\end{itemize}

The remainder of this paper is organized as follows. Section II introduces the PDNN-SSK system architecture, detailing the general model of the PDNN and its role in the SSK communication system. Section III provides the performance analysis, deriving the CCDP and the closed-form SER. Section IV formulates the PDNN optimization problem and presents the surrogate model-based training method to solve it. Section V provides our numerical results. Finally, Section VI concludes the paper and outlines directions for future research.

\textit{Notations}: Throughout this paper, we adopt the following notations. The set of complex and real numbers are denoted as $\mathbb{C}$ and $\mathbb{R}$, respectively. Scalars are denoted by italic letters, while vectors and matrices are represented by boldface lowercase (e.g., $\mathbf{v}$) and uppercase (e.g., $\mathbf{V}$) letters, respectively. $\mathbb{E}[\cdot]$ denotes the expectation operator. The operators $(\cdot)^{\mathrm{T}}$ and $(\cdot)^{\mathrm{H}}$ stand for the transpose and conjugate transpose. $|\cdot|$ denotes the magnitude of a complex scalar, while $\|\cdot\|$ is the Euclidean norm of a vector. $\operatorname{diag}(\mathbf{v})$ returns a diagonal matrix with the elements of $\mathbf{v}$ on its main diagonal, and $\operatorname{blkdiag}(\mathbf{V}_1, \cdots, \mathbf{V}_N)$ creates a block diagonal matrix from matrices $\mathbf{V}_1, \cdots, \mathbf{V}_N$. $\mathcal{CN}(\mu, \sigma^2)$ denotes a complex circularly-symmetric Gaussian distribution with mean $\mu$ and variance $\sigma^2$.

\begin{figure*}[ht!]
    \centering
    \subfloat[]{%
        \includegraphics[width=0.9\linewidth]{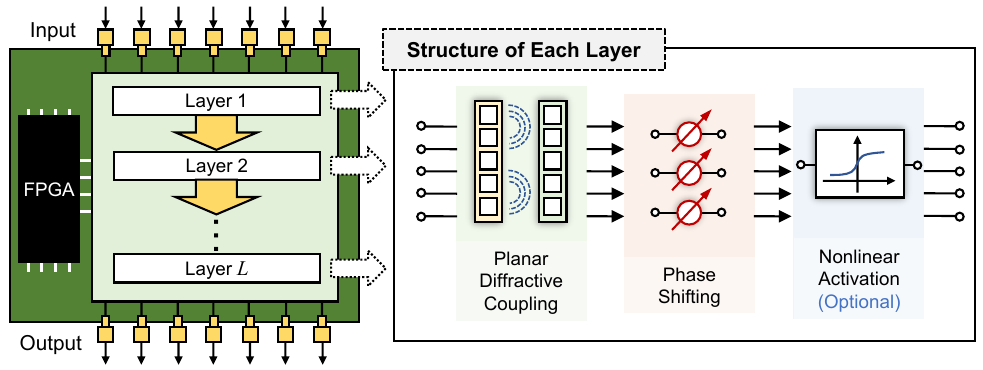}
        \label{fig:model_PDNN_a}
    }\hfill
    \subfloat[]{%
        \includegraphics[width=0.9\linewidth]{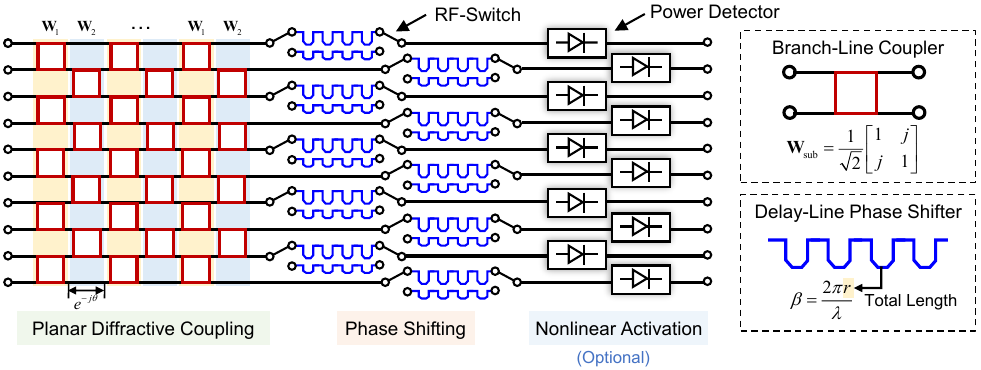}
        \label{fig:model_PDNN_b}
    }
	\vspace{-0.1cm}
    \caption{PDNN signal processing architecture. (a) General model of signal flow and layer structure. (b) Schematic diagram of an example layer design.}
    \label{fig:model_PDNN}
    \vspace{-0.4cm}
\end{figure*}

\section{The Proposed PDNN Empowered Communication System Model}
\label{s2}

\subsection{The General Model of PDNN}
\label{s2A}

To illustrate the concept of PDNN, we envisage a PCB with multiple input and output ports, fabricated with planar or flexible substrates \cite{guDirectElectromagneticInformation2024}. The RF signals propagate along the surface of PCB via artificially designed transmission circuits such as microstrip lines, while being processed by a multi-layer coupling and phase-shifting structure. Fig.~\ref{fig:model_PDNN}(a) depicts a general model of PDNN, which is composed of $L$ cascaded layers. Each layer consists of the following main modules: a planar diffractive coupling module, a phase-shifting module, and optionally, a nonlinear activation module to introduce nonlinearity. A field-programmable gate array (FPGA) is utilized to provide phase-shifting control signals.

To depict the signal model, we denote \( \mathbf{x}^{(0)} \in \mathbb{C}^{N^{(0)}} \) as the input signal to the entire PDNN, and let \( \mathbf{x}^{(l)} \in \mathbb{C}^{N^{(l)}} \) represent the output signal of the \( l \)-th layer, where \( N^{(l)} \) denotes the number of output ports of the \( l \)-th layer, \( l \in \mathcal{L} \) and \( \mathcal{L} = \{1, 2, \ldots, L\} \). The signal \( \mathbf{x}^{(l)} \) also serves as the input to the subsequent \( (l + 1) \)-th layer. In each layer, the signal sequentially passes through the aforementioned three functional modules. To characterize the coupling module, we introduce a transfer matrix \( \mathbf{W}^{(l)} \in \mathbb{C}^{N^{(l)} \times N^{(l-1)}} \). Notably, in the context of PDNNs, the transfer matrix \( \mathbf{W}^{(l)} \) captures the artificially engineered coupling effect that occurs as signals propagate along the PCB surface to emulate diffraction behavior. After that, the signal enters the phase-shifting module. We model the phase shift using a diagonal matrix \( \boldsymbol{\Phi}^{(l)} \in \mathbb{C}^{N^{(l)} \times N^{(l)}} \), whose diagonal entries are given by \( \phi_n^{(l)} = e^{j \beta_n^{(l)} } \) with \( \beta_n^{(l)} \) denoting the continuous phase shift applied at the \( n \)-th node of the \( l \)-th layer. These phase shift values in \( \boldsymbol{\Phi}^{(l)} \) are the main tunable parameters of the PDNN and are optimized to realize the desired mapping between input and output signals. Furthermore, to enhance the capability of PDNN, an optional nonlinear activation function \( f^{(l)} : \mathbb{C}^{N^{(l)}} \to \mathbb{C}^{N^{(l)}} \) can be applied after phase-shifting. Consequently, the effect of layer $l$ on the input signal can be expressed as
\begin{align}
	\mathbf{x}^{(l)} = f^{(l)}(\boldsymbol{\Phi}^{(l)} \mathbf{W}^{(l)} \mathbf{x}^{(l-1)}).
\end{align}
If no nonlinearity is applied, \( f^{(l)} \) defaults to the identity function, and the overall transformation of PDNN becomes
\begin{align}
	\label{eq:PDNN}
	\mathbf{x}^{(L)} = \boldsymbol{\Phi}^{(L)}\,\mathbf{W}^{(L)}\,\cdots\boldsymbol{\Phi}^{(1)}\,\mathbf{W}^{(1)} \mathbf{x}^{(0)}.
\end{align}

\subsection{An Example Design of PDNN layer}
\label{s2b}
Fig.~\ref{fig:model_PDNN}(b) depicts an example hardware design of the above PDNN layer model. This example illustrates how to implement a PDNN layer with available RF components. The planar diffractive coupling module is realized using a cascaded network composed of branch-line couplers \cite{guDirectElectromagneticInformation2024, gaoTerahertzSpoofPlasmonic2024}, which are 2-input and 2-output devices with a square structure. These couplers are printed on the PCB substrate and act as power splitters and combiners that distribute signal energy among multiple pathways. The characteristics of a basic component of these couplers can be mathematically described using a $2\times 2$ transfer matrix as
\begin{align}
	\mathbf{W}_\text{sub} = \frac{1}{\sqrt{2}} \begin{bmatrix} 1 & j \\ j & 1 \end{bmatrix},
\end{align}
which characterizes that each input port has an effect on each output port. When multiple branch-line couplers are arranged in parallel, the overall transfer matrix between the $N_\text{max}^{(l)}$ input ports and $N_\text{max}^{(l)}$ output ports can be constructed as
\begin{align}
	\mathbf{W}_1 = \operatorname{blkdiag}(\underbrace{ \mathbf{W}_\text{sub}, \dots, \mathbf{W}_\text{sub} }_{ N_\text{max}^{(l)}/2 }),
\end{align}
where $N_\text{max}^{(l)}$ is the smallest even integer greater than or equal to $\max\left\{ N^{(l)}, N^{(l-1)} \right\}$. Due to the special structure of this matrix, simply cascading $\mathbf{W}_{1}$ can only achieve coupling between each 2-input and 2-output port pairs, which is inadequate for establishing diffraction-like multi-port coupling. Therefore, another arrangement approach is required to construct an alternative transfer matrix. For instance, this matrix can be constructed as
\begin{align}
	\mathbf{W}_2 = \operatorname{blkdiag}(e^{-j\theta},\underbrace{ \mathbf{W}_\text{sub}, \dots, \mathbf{W}_\text{sub} }_{ N_\text{max}^{(l)}/2-1 },e^{-j\theta}),
\end{align}
where $\theta$ denotes the phase delay of the direct paths at the network boundaries, as illustrated in Fig.~\ref{fig:model_PDNN}(b). By cascading these structures for $M_\text{c}^{(l)}$ times, we can form a base transfer matrix for the layer-$l$:
\begin{align}
	\mathbf{W}_\text{base}^{(l)} = (\mathbf{W}_2 \mathbf{W}_1)^{M_\text{c}^{(l)}} \in \mathbb{C}^{N_\text{max}^{(l)} \times N_\text{max}^{(l)}},
\end{align}
which is a square matrix. To align with the general model in Section~\ref{s2A} where $\mathbf{W}^{(l)} \in \mathbb{C}^{N^{(l)} \times N^{(l-1)}}$ can be non-square, we first design a larger structure (i.e., the square matrix $\mathbf{W}_\text{base}^{(l)}$) and then obtain the actual matrix $\mathbf{W}^{(l)}$ by selecting a sub-matrix from it. In physical terms, this corresponds to fabricating or modeling the larger coupler network but only utilizing $N^{(l-1)}$ of its input ports and observing the signals at $N^{(l)}$ of its output ports. Furthermore, We note that the design approach for $\mathbf{W}^{(l)}$ is not unique. More diverse coupler selections and path planning can be incorporated into the design of the coupling module, leading to a variety of other transfer matrices, even with reconfigurability \cite{gaoProgrammableSurfacePlasmonic2023}.

After propagating through the branch-line couplers, the signal enters the phase-shifting module. As illustrated in Fig.~\ref{fig:model_PDNN}(b), this module is implemented using delay-line phase shifters, wherein the phase shift imparted to the signal directly corresponds to the lengths of microstrip delay lines. If we use $r$ to denote the total length of a certain signal path, the phase shift imparted to the signal can be expressed as $\beta = \frac{2\pi r}{\lambda}$, where $\lambda$ is the wavelength. These signal paths of different lengths can be meticulously designed to provide a set of discrete phase states by incorporating RF switches to switch between them \cite{rossaneseDesigningBuildingCharacterizing2022}. The low-cost RF switches and delay-line components render this architecture highly power-efficient \cite{wangReconfigurableIntelligentSurface2024}. Moreover, alternative phase shifter designs based on varactors allow for continuous phase tuning at each element, offering better flexibility in signal manipulation at the cost of more power consumption \cite{abboshCompactTunableReflection2012}. Accordingly, we maintain the assumption that the phase shift is continuously tunable for subsequent analysis and optimization, leaving discrete phase shift design for future research.

The realization of the optional nonlinear module is flexible, contingent upon the specific nonlinear function to be employed. In this work, we implement this nonlinearity using power detectors at the output ports of the last layer of PDNNs. These detectors convert the complex-valued wave signals into real-valued DC signals whose intensities are proportional to the signal power. This transformation allows the outputs of the PDNN, after normalization, to be interpreted as a probability distribution for subsequent tasks such as signal detection or image classification. While power detectors are typically exclusive to the final layer, other nonlinearities can be inserted into the other layers of PDNN to fulfill specific requirements \cite{gaoProgrammableSurfacePlasmonic2023}.

\begin{remark}
The proposed PDNN architecture offers compelling advantages for high-speed and energy-efficient wave-domain signal processing. Its inherently planar structure supports high spatial density, facilitating compact integration and high-throughput computation. When fabricated on a flexible substrate, PDNNs can conform to non-planar surfaces for versatile deployment. Meanwhile, unlike 3D diffraction systems suffering from significant propagation loss, the wave-guiding and coupling mechanism of PDNN ensures low propagation loss. It is thus feasible to construct deep PDNNs while maintaining high energy efficiency. Moreover, the planar design suppresses unpredictable near-field coupling and inter-layer reflection effects, resulting in higher fidelity between the fabricated PDNN and its theoretical model.
\end{remark}

\begin{figure*}[ht!]
	\centering
	\includegraphics[width=0.95 \textwidth]{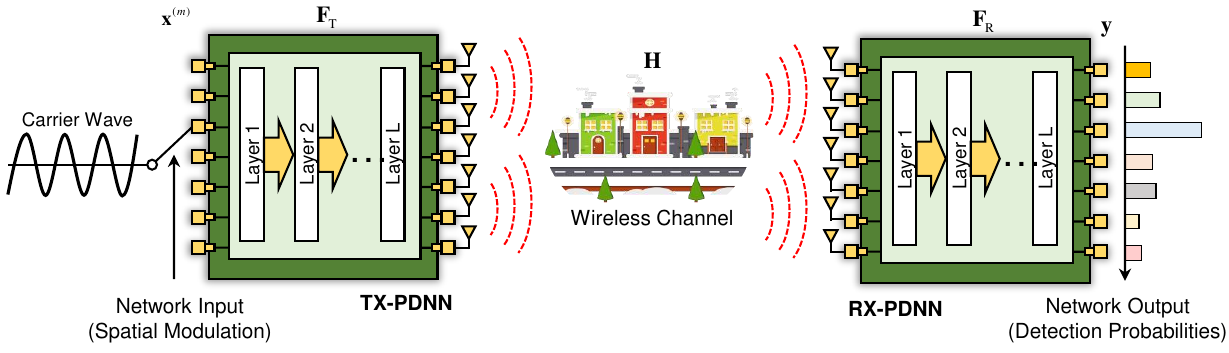}
	\vspace{-0.1cm}
	\caption{The schematic diagram of PDNN-SSK communication system}
	\label{fig:ASMMethod}
	\vspace{-0.4cm}
\end{figure*}

\subsection{PDNN-SSK communication system}

To fully harness the signal processing capabilities of PDNNs and simplify the transceiver architecture of conventional communication systems, we design a PDNN-SSK communication system with a single RF chain and a maximum power detector. As illustrated in Fig.~2, a PDNN is deployed at both the transmitter (TX) and receiver (RX) to serve as an electromagnetic signal processor. Only the final layer of the RX-PDNN deploys power detectors to provide nonlinearity, while all other layers in both the TX-PDNN and RX-PDNN do not include nonlinear modules. To distinguish the notations of TX-PDNN and RX-PDNN, we extend the previously defined symbols in Section \ref{s2A} by adding the subscripts ``\( \text{T} \)'' and ``\( \text{R} \)'', respectively. The number of layers for the TX-PDNN and RX-PDNN are denoted by \( L_\text{T} \) and \( L_\text{R} \), with the corresponding sets of layer indices given by \( \mathcal{L}_\text{T} = \{1, 2, \ldots, L_\text{T}\} \) and \( \mathcal{L}_\text{R} = \{1, 2, \ldots, L_\text{R}\} \), respectively. The transfer matrices are denoted as \( \mathbf{W}_\text{T}^{(l)} \in \mathbb{C}^{N_\text{T}^{(l)} \times N_\text{T}^{(l-1)}}, l \in \mathcal{L}_\text{T} \) and \( \mathbf{W}_\text{R}^{(l)} \in \mathbb{C}^{N_\text{R}^{(l)} \times N_\text{R}^{(l-1)}}, l \in \mathcal{L}_\text{R} \), where \( N_\text{T}^{(l)} \) and \( N_\text{R}^{(l)} \) are the number of ports in the PDNN layer $l$ at the transmitter and the receiver, respectively. The phase shift matrices are represented by \( \boldsymbol{\Phi}_\text{T}^{(l)} \in \mathbb{C}^{N_\text{T}^{(l)} \times N_\text{T}^{(l)}} \) and \( \boldsymbol{\Phi}_\text{R}^{(l)} \in \mathbb{C}^{N_\text{R}^{(l)} \times N_\text{R}^{(l)}} \). Other related notations are extended following the same rule.

At the transmitter, only one RF chain is required to emit a pure carrier wave, and the SSK modulation is implemented by selecting a specific PDNN port to feed the signal. Let $\mathbf{s}\in\mathbb{C}^{M}$ denote the one-hot port activation vector as well as the transmitted signal vector, where $M=2^{p}\in\mathbb{N}^{+}$ represents the modulation order and $p$ denotes the number of transmitted bits. When the $m$-th port is activated, $\mathbf{s}$ can be expressed by
\begin{align}
	\mathbf{s} = [\underbrace{0, \cdots, 0}_{m-1}, \underbrace{1}_{m\text{-th entry}}, \underbrace{0, \cdots, 0}_{M-m}]^T,  \quad m\in \mathcal{M},
\end{align}
where $\mathcal{M}=\left\{ 1,2,\dots,M \right\}$. With \( \mathbf{x}^{(0)}_\text{T} \in \mathbb{C}^{N^{(0)}_\text{T}} \) denoting the input signal to the entire TX-PDNN, we have $\mathbf{x}^{(0)}_\text{T}=\mathbf{s}$ and $N^{(0)}_\text{T}=M$.

After being processed by the TX-PDNN, the signals are radiated from the transmit antenna array connected to the TX-PDNN output ports. These signals then propagate through the wireless channel and are captured by the receive antenna array before being fed into the RX-PDNN for further processing. For this end-to-end transmission, the only nonlinearity considered in our proposed system comes from the power detectors in the final layer of the RX-PDNN. The received signal $\mathbf{y}\in \mathbb{R}^{M}$ at the the $M = N^{(L_\text{R})}_\text{R}$ output ports of RX-PDNN can thus be expressed as
\begin{align}
	\mathbf{y}=|\mathbf{F}_{\mathrm{\text{R}}}\mathbf{H}\,\mathbf{F}_{\mathrm{\text{T}}}\mathbf{s}+\mathbf{n}|^2,
\end{align}
where $\mathbf{H}\in\mathbb{C}^{N_\text{R}^{(0)} \times N_\text{T}^{(L_\text{T})}}$ stands for the Rayleigh-fading channel matrix with independent and identically distributed (i.i.d.) entries following $\mathcal{CN}(0,1)$, and the additive noise $\mathbf{n}\in\mathbb{C}^{M}$ follows $\mathcal{CN}(0,2\sigma^2\mathbf{I})$. $\mathbf{F}_{\mathrm{\text{T}}}$ and $\mathbf{F}_{\mathrm{\text{R}}}$ denote the transfer matrices of the entire TX-PDNN and RX-PDNN, respectively. According to \eqref{eq:PDNN}, they are defined as
\begin{align}
	\mathbf{F}_{\mathrm{\text{T}}}=\boldsymbol{\Phi}_{\text{T}}^{(L_{\text{T}})}\,\mathbf{W}_{\text{T}}^{(L_{\text{T}})}\,\cdots\boldsymbol{\Phi}_{\text{T}}^{(1)}\,\mathbf{W}_{\text{T}}^{(1)},\\
	\mathbf{F}_{\mathrm{\text{R}}}=\mathbf{W}_{\text{R}}^{(L_{\text{R}})}\,\boldsymbol{\Phi}_{\text{R}}^{(L_{\text{R}})}\,\cdots\mathbf{W}_{\text{R}}^{(1)}\,\boldsymbol{\Phi}_{\text{R}}^{(1)},
\end{align}
where in the RX-PDNN, the order of the planar diffractive coupling module and the phase modulation module is reversed.\footnote{The reason for this reversal is that if the final module in the RX-PDNN were a phase modulation module, its effect would be nullified by the subsequent power detection.}

At the receiver, the signal detection is carried out in a non-coherent manner, relying solely on the power levels at the output ports of the RX-PDNN. The index of the output port with the maximum power level is selected and mapped to the corresponding transmitted symbol, which is expressed as
\begin{align}
	\hat{m} = \arg\max_{m} y_m,
\end{align}
where \( \hat{m} \) denotes the estimated symbol index. 

\section{Performance Analysis}
\label{sec:performance}
In this section, we analyze the performance of the proposed PDNN-SSK communication system when the receiver employs a non-coherent detector. This detector selects the symbol corresponding to the output port with the maximum signal magnitude---a method that simplifies receiver hardware at the cost of being suboptimal compared to a coherent maximum likelihood (ML) detector. Our objective is to derive the conditions under which this detection scheme achieves its optimal performance and provide closed-form SER expressions.

\subsection{Conditional Correct Detection Probability}
When the symbol $\mathbf{s}_m$ is transmitted, the detector makes a correct decision if the signal power at the $m$-th output port is greater than that at any other port. The CCDP, given the channel matrix $\mathbf{H}$, is defined as
\begin{align}
	\label{eq:pcm_ini}
	P_{c,m} &= P(\hat{m}=m \mid \mathbf{s}_m, \mathbf{H}) \nonumber \\
	&= P\left( |y_m| > |y_{m'}|, \forall m' \in \mathcal{M}, m' \neq m \mid \mathbf{s}_m, \mathbf{H} \right),
\end{align}
where comparing signal powers is equivalent to comparing their amplitudes in this context. We define the effective channel coefficient from the $m$-th transmitter port to the $m'$-th receiver port as 
\begin{align}
    c_{m',m} = \mathbf{f}_{\text{R},m'}^{\mathrm{T}}\mathbf{HF}_\text{T}\mathbf{s}_m,
\end{align}
where $\mathbf{f}_{\text{R},m'}^{\mathrm{T}}$ denotes the $m'$-th row of $\mathbf{F}_\mathrm{R}$. The received signal at the $m'$-th output port thus can be expressed as $y_{m'} = c_{m',m} + n_{m'}$. Accordingly, the CCDP expression \eqref{eq:pcm_ini} can be recast into
\begin{align}
	\label{eq:pcm_2}
	P_{c,m} = P\bigg( \bigcap_{\substack{m' \in \mathcal{M} \\ m' \neq m}} \bigg\{ |c_{m,m} + n_m| > |c_{m',m} + n_{m'}| \bigg\} \,\bigg|\, \mathbf{s}_m, \mathbf{H} \bigg).
\end{align}
Let $X_{m'} = |y_{m'}| = |c_{m',m} + n_{m'}|$. Since the noise term $n_{m'} \sim \mathcal{CN}(0, 2\sigma^2)$, the random variable $X_{m'}$ follows a Rician distribution. By conditioning on the value of $X_m$ and leveraging the independence of the noise terms across different ports, we transform \eqref{eq:pcm_2} into
\begin{equation}
	\label{eq:pcm_integral}
	\begin{aligned}
		P_{c,m}=& \int_0^{\infty} f_{X_m}(r; |c_{m,m}|) \\
		&\times \Bigg( \prod_{\substack{m' \in \mathcal{M} \\ m' \neq m}} P\left( X_{m'} < r  \mid \mathbf{s}_m, \mathbf{H} \right) \Bigg) dr,
	\end{aligned}
\end{equation}
where $f_{X_m}(r; \mu)$ is the probability density function (PDF) of a Rician random variable with non-centrality parameter $\mu = |c_{m,m}|$, given by
\begin{align}
	f_{X_m}(r; \mu) = \frac{r}{\sigma^2} \exp\left(-\frac{r^2+\mu^2}{2\sigma^2}\right) I_0\left(\frac{r\mu}{\sigma^2}\right), \, r \ge 0,
\end{align}
and $I_0(\cdot)$ is the modified Bessel function of the first kind and zero order. The corresponding cumulative distribution function (CDF) is
\begin{align}
	F_{X_{m'}}(r; \mu) = 1 - Q_1\left(\frac{\mu}{\sigma}, \frac{r}{\sigma}\right),
\end{align}
where $Q_1(\cdot, \cdot)$ is the first-order Marcum Q-function \cite{marcumStatisticalTheoryTarget1960}. 

To determine the optimal conditions for detection, we establish the following theorems.
\begin{theorem}[Monotonicity of CCDP]
	\label{thm:monotonicity}
	The CCDP $P_{c,m}$ is monotonically increasing with the desired signal amplitude, $|c_{m,m}|$, and monotonically decreasing with each of the interfering signal amplitudes, $|c_{m',m}|$, for all $m' \neq m$.
\end{theorem}

\begin{proof}
	First, we establish the monotonicity of $P_{c,m}$ with respect to an interfering signal amplitude $|c_{m',m}|$, where $m' \neq m$. The value of $P_{c,m}$ in \eqref{eq:pcm_integral} is dependent upon the product of CDF terms 
	\begin{equation}
		\begin{aligned}
			P\left( X_{m'} < r  \mid \mathbf{s}_m, \mathbf{H} \right) &= F_{X_{m'}}(r; |c_{m',m}|) \\
			&= 1 - Q_1 \left(\frac{|c_{m',m}|}{\sigma}, \frac{r}{\sigma} \right),
		\end{aligned}
	\end{equation}
	in which the Marcum Q-function $Q_1(a,b)$ is strictly increasing with respect to its first argument $a$ for $a,b > 0$ \cite{sunMonotonicityLogConcavityTight2010}. Consequently, $P\left( X_{m'} < r  \mid \mathbf{s}_m, \mathbf{H} \right)$ is a strictly decreasing function of $|c_{m',m}|$. As all terms within the integrand of \eqref{eq:pcm_integral} are non-negative, it directly follows that $P_{c,m}$ is a strictly decreasing function of $|c_{m',m}|$, for all $m' \neq m$.

	Next, we prove that $P_{c,m}$ is monotonically increasing with $|c_{m,m}|$, $\forall m \in \mathcal{M}$. Let $X_\alpha=|c_{\alpha}+n_{\alpha}|$ and $X_\beta=|c_{\beta}+n_{\beta}|$ denote the received signal amplitude corresponding to noncentral parameters $c_\alpha$ and $c_\beta$, respectively, with $0 < |c_\alpha| < |c_\beta|$. Since \eqref{eq:pcm_integral} can be regarded as the expectation of a function of random variable $X_m$, which can be defined as
	\begin{align}
		g(X_m) = \prod_{m' \neq m} F_{X_{m'}}(X_m; |c_{m',m}|),
	\end{align}
	it is equivalent to prove that $\mathbb{E}[g(X_{\alpha})] < \mathbb{E}[g(X_{\beta})]$. According to the aforementioned monotonicity of the Marcum Q-function, it can be readily shown that $g(X_m)$ is a strictly increasing function for $X_m \ge 0$. We define the difference in expectation as
	\begin{equation}
		\begin{aligned}
			\Delta &= \mathbb{E}[g(X_{\beta})] - \mathbb{E}[g(X_{\alpha})] \\
			&= \int_0^\infty \left(f_{X_\beta}(r; |c_\beta|) - f_{X_\alpha}(r; |c_\alpha|)\right) g(r) dr.
		\end{aligned}
	\end{equation}
	Applying integration by parts yields
	\begin{equation}
		\begin{aligned}
			\Delta =& \left[ g(r) \left( F_{X_\beta}(r;|c_\beta|) - F_{X_\alpha}(r;|c_\alpha|) \right) \right]_0^\infty  \\
			&- \int_0^\infty \left( F_{X_\beta}(r;|c_\beta|) - F_{X_\alpha}(r;|c_\alpha|) \right) g'(r) dr.
		\end{aligned}
	\end{equation}
	The boundary term is zero, since $F_{X_{\alpha}}(0;|c_{\alpha}|) = F_{X_{\beta}}(0;|c_{\beta}|) = 0$ and $F_{X_{\alpha}}(+\infty;|c_{\alpha}|) = F_{X_{\beta}}(+\infty;|c_{\beta}|) = 1$. The expression thus reduces to
	\begin{align}
		\label{eq:inta1}
		\Delta = \int_0^\infty \left( F_{X_\alpha}(r; |c_\alpha|) - F_{X_\beta}(r; |c_\beta|) \right) g'(r) dr.
	\end{align}
	Given that $|c_\alpha| < |c_\beta|$, we have 
	\begin{equation}
		\begin{aligned}
			F_{X_{\alpha}}(r;|c_{\alpha}|) =& 1-Q_1\left( \frac{|c_{\alpha}|}{\sigma}, \frac{r}{\sigma}\right) \\
			>& 1-Q_1\left( \frac{|c_{\beta}|}{\sigma}, \frac{r}{\sigma}\right) = F_{X_{\beta}}(r;|c_{\beta}|).
		\end{aligned}
	\end{equation}
	Meanwhile, $g'(r) > 0$ because $g(r)$ is strictly increasing. The integrand in \eqref{eq:inta1} is therefore strictly positive for $r > 0$, which ensures that $\Delta > 0$, i.e., 
	\begin{align}
		\mathbb{E}[g(X_{\alpha})] < \mathbb{E}[g(X_{\beta})], \text{ if }|c_\alpha| < |c_\beta|.
	\end{align}
	This confirms that $P_{c,m}$ is strictly increasing with $|c_{m,m}|$.
\end{proof}

\begin{proposition}[Optimal Condition for CCDP]
	\label{thm:optimal_cond}
	$P_{c,m}$ is maximized only if
	\begin{align}
		c_{m',m} = 0, \quad \forall m' \in \mathcal{M}, m' \neq m.
	\end{align}
	Under this optimal condition, the CCDP achieves its maximum value, $P_{c,m}^*$, given by
	\begin{align}
		\label{eq:pcm_optimal}
		P_{c,m}^* = \int_0^{\infty} f_{X_m}(r; |c_{m,m}|) \left( 1-e^{-\frac{r^2}{2\sigma^2}} \right)^{M-1} dr.
	\end{align}
\end{proposition}

\begin{proof}
	This proposition is a direct consequence of Theorem~\ref{thm:monotonicity}. Since $P_{c,m}$ is a strictly decreasing function of each interference amplitude $|c_{m',m}|$ for $m' \neq m$, the maximization of $P_{c,m}$ is achieved when $|c_{m',m}| = 0$, $\forall m' \neq m$. In this scenario, the signals at the unintended output ports consist solely of noise, i.e., $y_{m'} = n_{m'}$. The amplitude $|y_{m'}|$ is thus a Rayleigh-distributed random variable, whose CDF becomes
	\begin{equation}
		\label{eq:rayleighcdf}
		\begin{aligned}
			F_{X_{m'}}(r; 0) &= 1 - Q_1\left(0, \frac{r}{\sigma}\right) \\
			&= 1 - e^{-\frac{r^2}{2\sigma^2}}. 
		\end{aligned}
	\end{equation}
	Substituting \eqref{eq:rayleighcdf} into \eqref{eq:pcm_integral} yields the expression for the maximum achievable CCDP, $P_{c,m}^*$, as given in \eqref{eq:pcm_optimal}.
\end{proof}

\begin{remark}
    Theorem \ref{thm:monotonicity} and Proposition \ref{thm:optimal_cond} provide a guideline for the design of the PDNN. That is, the beamforming matrices $\mathbf{F}_{\text{T}}$ and $\mathbf{F}_{\text{R}}$ should collaborate to null the signal path to all unintended receiver ports and maximize the signal magnitude at the intended receiver port, which are essential for robust non-coherent detection.
\end{remark}

\subsection{Closed-Form SER Analysis}
In this subsection, we derive some accurate and asymptotic closed-form expressions for SER analysis of the proposed PDNN-SSK communication system. The following theorem provides an accurate close-form SER expression under the optimal condition for detection given in Proposition~\ref{thm:optimal_cond}.

\begin{theorem}[Accurate closed-form SER]
	\label{thm:closed_form_ser}
	Under the optimal detection condition where inter-port interference is completely eliminated (i.e., $c_{m',m} = 0, \forall m' \in \mathcal{M}, m' \neq m$), the SER for transmitting symbol $\mathbf{s}_m$, denoted as $\text{SER}_m$, is given by the closed-form expression:
	\begin{align}
		\label{eq:ser_closed_form}
		\text{SER}_m = \sum_{k=1}^{M-1} \binom{M-1}{k} \frac{(-1)^{k+1}}{k+1} \exp\left(-\frac{\gamma_m k}{k+1}\right),
	\end{align}
	where $\gamma_m = |c_{m,m}|^2 / (2\sigma^2)$ is the effective signal-to-noise ratio (SNR) at the intended output port $m$.
\end{theorem}

\begin{proof}
	We start with applying the binomial expansion to the binomial term in the CCDP expression given by \eqref{eq:pcm_optimal}:
	\begin{align}
		\left( 1 - e^{-\frac{r^2}{2\sigma^2}} \right)^{M-1} = \sum_{k=0}^{M-1} \binom{M-1}{k} (-1)^k e^{-\frac{k r^2}{2\sigma^2}}.
	\end{align}
	Substituting this into \eqref{eq:pcm_optimal} and interchanging the order of summation and integration, which is permissible due to the absolute convergence of the integral for each term, we obtain:
	\begin{align}
		\label{eq:pc_sum_integral}
		P_{c,m}^* = \sum_{k=0}^{M-1} \binom{M-1}{k} (-1)^k \int_0^{\infty} f_{X_m}(r; |c_{m,m}|) e^{-\frac{k r^2}{2\sigma^2}} dr.
	\end{align}
	Let's evaluate the integral term, which we denote as $\mathcal{I}_k$:
	\begin{align}
        \label{eq:I_k_1}
		\mathcal{I}_k &= \int_0^{\infty} f_{X_m}(r; |c_{m,m}|) e^{-\frac{k r^2}{2\sigma^2}} dr \nonumber \\
		&= \int_0^{\infty} \frac{r}{\sigma^2} \exp\left(-\frac{r^2+|c_{m,m}|^2}{2\sigma^2}\right) I_0\left(\frac{r|c_{m,m}|}{\sigma^2}\right) e^{-\frac{k r^2}{2\sigma^2}} dr \nonumber \\
		&= \frac{1}{\sigma^2}\exp\left(-\frac{|c_{m,m}|^2}{2\sigma^2}\right) \nonumber \\
		&\quad \times \int_0^{\infty} r \exp\left(-\frac{(k+1)r^2}{2\sigma^2}\right) I_0\left(\frac{r|c_{m,m}|}{\sigma^2}\right) dr.
	\end{align}
	Using the standard integral identity from \cite[Eq. 6.631.4]{gradshteinTableIntegralsSeries2014},
	\begin{align}
		\int_0^\infty x \exp(-ax^2) I_0(bx) dx = \frac{1}{2a} \exp\left(\frac{b^2}{4a}\right).
	\end{align}
	Setting $a = \frac{k+1}{2\sigma^2}$, $b = \frac{|c_{m,m}|}{\sigma^2}$ and substituting these into \eqref{eq:I_k_1} yields:
	\begin{align}
		\mathcal{I}_k &= \frac{1}{k+1} \exp\left(-\frac{|c_{m,m}|^2}{2\sigma^2}\right) \exp\left(\frac{|c_{m,m}|^2}{2\sigma^2(k+1)}\right).
	\end{align}
	By defining the effective SNR as $\gamma_m = |c_{m,m}|^2 / (2\sigma^2)$, we simplify $\mathcal{I}_k$ to:
	\begin{align}
		\mathcal{I}_k = \frac{1}{k+1} \exp\left(-\frac{\gamma_m k}{k+1}\right).
	\end{align}
	Substituting this result back into \eqref{eq:pc_sum_integral}, the correct detection probability is:
	\begin{align}
		P_{c,m}^* = \sum_{k=0}^{M-1} \binom{M-1}{k} \frac{(-1)^k}{k+1} \exp\left(-\frac{\gamma_m k}{k+1}\right).
	\end{align}
	The SER for symbol $m$ is denoted as $\text{SER}_m = 1 - P_{c,m}^*$. Therefore, we have
	\begin{align}
		\text{SER}_m = \sum_{k=1}^{M-1} \binom{M-1}{k} \frac{(-1)^{k+1}}{k+1} \exp\left(-\frac{\gamma_m k}{k+1}\right).
	\end{align}
	This completes the proof.
\end{proof}

\begin{remark}
	The SER expression \eqref{eq:ser_closed_form} is identical in form to that of conventional $M$-ary FSK with non-coherent detection \cite{proakisDigitalCommunications2008}. This correspondence is not a coincidence; it reveals that under the optimal interference-nulling condition, our PDNN-SSK modulation scheme behaves as a set of $M$ orthogonal channels, which is similar to the $M$-ary FSK modulation.
\end{remark}

The SER expression in \eqref{eq:ser_closed_form} is a weighted sum of exponential terms of the form $\exp(-\frac{\gamma_m k}{k+1})$. The function $g(k) = k/(k+1)$ is monotonically increasing for $k \ge 1$. Consequently, as $\gamma_m \to \infty$, the exponential term with the smallest value of $g(k)$ will dominate the sum. This occurs at $k=1$, where $g(1)=1/2$. The subsequent terms, for $k \ge 2$, decay much more rapidly. Therefore, for large $\gamma_m$, the $k=1$ term provides a precise asymptotic expression for the SER. Based on the above observation, the following proposition provides an asymptotic approximation to~\eqref{eq:ser_closed_form}. 

\begin{proposition}[Asymptotic closed-form SER]
	\label{cor:asymptotic_ser}
	For large SNR $\gamma_m$, the SER in \eqref{eq:ser_closed_form} is dominated by the most significant term in the summation (corresponding to $k=1$) and can be tightly approximated by:
	\begin{align}
		\label{eq:ser_asymptotic}
		\text{SER}_m \approx \frac{M-1}{2} \exp\left(-\frac{\gamma_m}{2}\right).
	\end{align}
\end{proposition}

\begin{proof}
    The proof can be completed by simply truncating the summation in \eqref{eq:ser_closed_form} while retaining the term with $k=1$.
\end{proof}

\begin{remark}
	The asymptotic expression in Proposition \ref{cor:asymptotic_ser} provides critical insights into the system's scalability. A key observation is how the modulation order, $M$, impacts the SER. Notably, $M$ only affects the leading multiplicative constant, $(M-1)/2$, while the exponential term remains independent of $M$. This behavior is different from many common modulation schemes, such as M-QAM or M-PSK. In those schemes, increasing the modulation order $M$ directly reduces the multiplier on SNR within the exponential term and causes a SNR penalty. In contrast, the PDNN-SSK modulation scheme demonstrates a more graceful trade-off between the modulation order and energy efficiency.
\end{remark}

\section{PDNN Optimization for Optimal Detection}

\subsection{Problem Formulation}
\label{sec:problem_formulation}

Theorems \ref{thm:monotonicity} and Proposition \ref{thm:optimal_cond} in Section \ref{sec:performance} provide a clear directive for minimizing the symbol error rate: For any transmitted symbol $\mathbf{s}_m$, the system should be configured to maximize the signal power at the intended output port $m$ while simultaneously suppressing the interference power at all other ports $m' \neq m$. This principle serves as an objective for the optimal design of PDNN parameters. Accordingly, we define a symbol-wise signal-to-interference-plus-noise ratio (SINR), which quantifies the ratio of the desired signal power to the combined power of all interference signals and thermal noise:
\begin{align}
	\label{eq:sinr_def}
	\text{SINR}_{m} \triangleq \frac{|c_{m,m}|^{2}}
	{\displaystyle \sum_{m'\neq m}|c_{m',m}|^{2}+\sigma_{n}^{2}},
\end{align}
where $\sigma_{n}^{2}=2\sigma^2$ represents the effective noise variance. When the interference term $|c_{m',m}|$ is extremely small—potentially even below the noise level—the optimization objective should shift from further suppressing interference to enhancing the desired signal $|c_{m,m}|$. To reflect this behavior, we incorporate the effect of noise into the optimization objective. 

With inter-symbol interference existing, directly maximizing the CCDP in \eqref{eq:pcm_integral} is mathematically intractable due to the complex integral form. Instead, we formulate the phase shift optimization problem of PDNN by maximizing the effective achievable sum-rate as follows:
\begin{subequations}
	\label{eq:opt_problem}
	\begin{align}
		\max_{\{\boldsymbol{\Phi}_{\mathrm{T}}^{(l)}\}, \{\boldsymbol{\Phi}_{\mathrm{R}}^{(l)}\}} \quad & \sum_{m=1}^{M}\log_2(1+\text{SINR}_{m}) \label{eq:opt_obj}, \\
		\mathrm{s.t.} \quad\quad & |[\boldsymbol{\Phi}_{\mathrm{T}}^{(l)}]_{n,n}|=1, \; \forall n, \forall l \in \mathcal{L}_\text{T}, \label{eq:opt_const_tx} \\
		& |[\boldsymbol{\Phi}_{\mathrm{R}}^{(l)}]_{n,n}|=1, \; \forall n, \forall l \in \mathcal{L}_\text{R}. \label{eq:opt_const_rx}
	\end{align}
\end{subequations}
This problem is highly non-convex and thus intractable to solve, primarily due to the inherent cascaded architecture of the PDNNs. Prior work has addressed similar problems in multi-user MISO scenarios using projected gradient ascent methods with Armijo step size \cite{anStackedIntelligentMetasurfaces2025}. However, such approaches are prone to getting trapped in local optima, resulting in substantial performance gaps from the global optimum. To overcome this limitation, we adopt a surrogate model-based training strategy in this work, leveraging the powerful global optimization capability of machine learning optimizers. This approach allows automatic differentiation and adaptive gradient-based optimization of the effective achievable sum-rate.

\begin{remark}
    The choice of an appropriate metric to quantify the combined impact of noise and interference in such scenarios warrants further discussion. In this work, we adopt a simple yet effective approach, where the total power of noise and interference is combined additively.
\end{remark}

\subsection{Surrogate Model Training}

To solve the optimization problem in \eqref{eq:opt_problem}, we reframe this task as training a surrogate model that mirrors the PDNN-SSK communication system. In this paradigm, the signal evolution through the TX-PDNN, the wireless channel, and the RX-PDNN is represented as a sequence of differentiable layer operations. The tunable phase shifts $\{\beta_{n,\text{T}}^{(l)}\}$ and $\{\beta_{n,\text{R}}^{(l)}\}$, which form the diagonal elements of the optimization matrices $\{\boldsymbol{\Phi}_{\mathrm{T}}^{(l)}\}$ and $\{\boldsymbol{\Phi}_{\mathrm{R}}^{(l)}\}$ in \eqref{eq:opt_problem}, are treated as the trainable parameters of this model.

The forward propagation process of the surrogate model is firstly to compute the entire end-to-end effective channel matrix $\mathbf{C} = \mathbf{F}_\text{R} \mathbf{H} \mathbf{F}_\text{T}$, with $c_{m',m}$ as its $m'$-th row and $m$-th column element. Subsequently, all $M$ SINRs are calculated to obtain the loss function value. To maximize the expression in \eqref{eq:opt_obj}, the loss function for minimization is designed as $ \mathcal{J}=-\sum_{m=1}^{M}\log(1+\text{SINR}_{m})$ in the training framework. By converting the whole system into a surrogate model, we can employ the back-propagation algorithm to acquire the gradients of the loss function with respect to all trainable phase parameters.

\begin{algorithm}[bp]
    \caption{PDNN Training via Surrogate Model}
    \label{alg:training}
    \begin{algorithmic}[1]
        \STATE \textbf{Input:} Channel matrix $\mathbf{H}$, fixed diffractive matrices $\{\mathbf{W}_\text{T}^{(l)}\}, \{\mathbf{W}_\text{R}^{(l)}\}$, modulation order $M$, noise power $\sigma_n^{2}$.
        
        \STATE Randomly initialize all phase parameters $\{\beta_{n,\text{T}}^{(l)}\}, \{\beta_{n,\text{R}}^{(l)}\}$.
        \STATE Initialize Adam optimizer with a learning rate $\eta$.
        \FOR{epoch = 1 to $N_{\text{epochs}}$}
        \STATE Compute the end-to-end effective channel matrix $\mathbf{C} = \mathbf{F}_\text{R} \mathbf{H} \mathbf{F}_\text{T}$.
        \STATE Calculate $\{\text{SINR}_m\}_{m=1}^M$ from $\mathbf{C}$ via \eqref{eq:sinr_def}.
        \STATE Calculate the total loss $\mathcal{J} = -\sum_{m=1}^{M}\log(1+\text{SINR}_{m})$.
        \STATE Compute gradients of $\mathcal{J}$ with respect to all phase parameters $\{\beta_{n,\text{T}}^{(l)}\}, \{\beta_{n,\text{R}}^{(l)}\}$ using back-propagation.
        \STATE Update all phase shifts using the Adam optimizer step.
        \ENDFOR
        \STATE \textbf{Output:} Optimized $\{\boldsymbol{\Phi}_{\mathrm{T}}^{(l)*}\}, \{\boldsymbol{\Phi}_{\mathrm{R}}^{(l)*}\}$.
    \end{algorithmic}
\end{algorithm}

For the parameter update step, we utilize the Adam optimizer with a preset learning rate $\eta$ \cite{kingmaAdamMethodStochastic2017}, which incorporates both first-order momentum and second-order momentum with adaptive per-parameter scaling to dynamically adjust the learning rate for each phase shifter and effectively dampen oscillations, accelerate convergence through flat regions, and escape local minima. The overall optimization procedure is detailed in Algorithm~\ref{alg:training}.

\begin{figure*}[ht!]
    \centering
    \includegraphics[width=1 \textwidth]{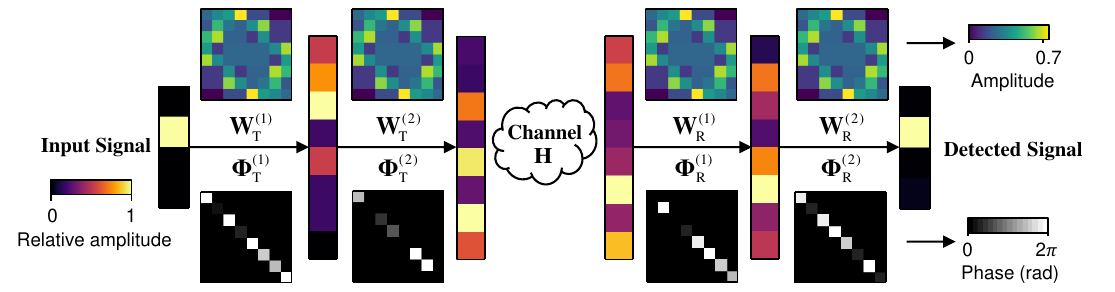}
    \vspace{-0.1cm}
    \caption{Visualization of the wave-based signal processing in the PDNN-SSK communication system with $L_T=L_R=2$. The vertical strips are 1D heatmaps of the signal amplitude, while the square and diagonal heatmaps illustrate the coupling amplitude matrices and the trained phase shift matrices, respectively.}
    \label{fig:EMwaveProp}
\end{figure*}

\section{Numerical Results}
In this section, we present numerical results to validate the performance of the proposed PDNN-SSK communication system and the associated optimization algorithm.

\subsection{Simulation Setup}
\label{sec:sim_setup}
Unless otherwise specified, the following setup is adopted for all simulations. The fixed diffractive coupling matrices, $\mathbf{W}^{(l)}$, of the PDNN layers are constructed based on the cascaded branch-line coupler design detailed in Section \ref{s2b}, with the phase delay of the direct paths set as $\theta = \pi/4$. All the layers of TX-PDNN and RX-PDNN have the same number of ports, except for the input and output layers, i.e. $N^{(l_\text{T})}_\text{T}=N^{(l_\text{R})}_\text{R}=N$, where $l_\text{T}=1,\dots,L_\text{T}, l_\text{R}=0,\dots,L_\text{R}-1$. For each independent channel realization $\mathbf{H}$, the TX and RX-PDNNs are jointly trained to optimize the effective achievable sum-rate defined in \eqref{eq:opt_obj}. The optimization is performed using the gradient-based algorithm detailed in Algorithm~\ref{alg:training}. The phase shift parameters of all PDNN layers are randomly initialized from a uniform distribution $\mathcal{U}(-\pi, \pi)$ before training begins.

\subsection{EM Wave Propagation}

To visually elucidate the fundamental principle of wave-based signal processing within the proposed PDNN architecture, Fig.~\ref{fig:EMwaveProp} illustrates the propagation of a signal through the end-to-end PDNN-SSK communication system. The transmitter and receiver are each equipped with a 2-layer PDNN ($L_\text{T}=L_\text{R}=2$), each layer comprising $N=8$ ports, for a system supporting modulation order $M=4$. The cascade factor is set to $M_\text{c}^{(l)}=3$ for all $\mathbf{W}_\text{T}^{(l)}$ and $\mathbf{W}_\text{R}^{(l)}$, and their corresponding heatmaps confirm that this configuration establishes strong coupling among most of ports, which is essential for the network's computational capability. The process begins with a single-port excitation at the input. As this signal propagates through the cascaded layers of the transmitting PDNN, the optimized phase shifts from $\mathbf{\Phi}_\text{T}^{(1)}$ and $\mathbf{\Phi}_\text{T}^{(2)}$ and the fixed coupling of $\mathbf{W}_\text{T}^{(1)}$ and $\mathbf{W}_\text{T}^{(2)}$ work in concert to transform the simple input into a high-dimensional wavefront. After traversing the wireless channel $\mathbf{H}$, the receiving PDNN performs the inverse operation: its layers progressively gather the diffused energy, culminating in the precise focusing of the signal's power onto the correct output port corresponding to the initial input. This entire sequence of transformations is not executed by conventional digital processors but occurs inherently as the signal propagates through the PDNNs at the speed of light, thereby offloading the corresponding computational burden of the digital baseband module.

\begin{figure}[t]
    \centering
    \includegraphics[width=0.5\textwidth]{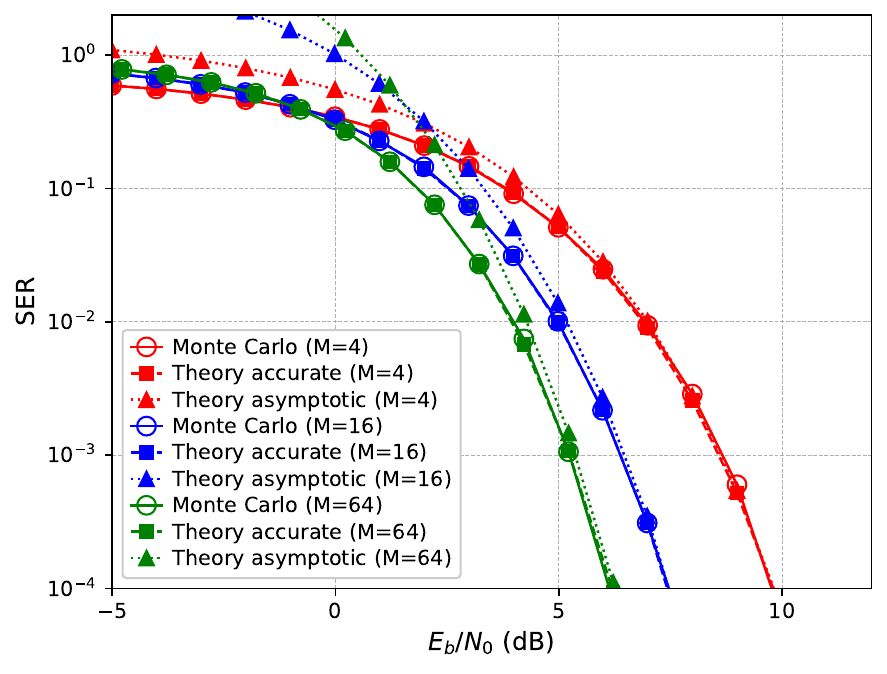}
    \vspace{-0.6cm}
    \caption{Theoretical and simulated SER performance for modulation orders of $M\in \{4, 16, 64\}$}
    \label{fig:theory}
    \vspace{-0.4cm}
\end{figure}

\subsection{SER Performance}

To validate our analytical framework for SER presented in Section III, we conduct Monte Carlo simulations, with the results depicted in Fig.~\ref{fig:theory}. The simulation environment is configured in the ideal, interference-free scenario assumed in Theorem \ref{thm:closed_form_ser} and Proposition \ref{cor:asymptotic_ser}. The SER performance is evaluated for modulation orders of $M\in \{4, 16, 64\}$, with each data point averaged over $10^5$ independent realizations to ensure statistical significance. As depicted in Fig.~\ref{fig:theory}, the $x$-axis represents the SNR per bit, $E_\text{b}/N_0$, which is derived from $\gamma_m$ according to the relation $E_\text{b}/N_0=\gamma_m/(2\log_2M)$. The simulated SER results exhibit a perfect match with the corresponding theoretical curves derived from the accurate closed-form SER expression in \eqref{eq:ser_closed_form}. Meanwhile, the utility of the asymptotic SER expression \eqref{eq:ser_asymptotic} is demonstrated; its corresponding theoretical curves rapidly converge to the accurate SER curve in the high-SNR regime, providing a tight upper bound, particularly for SER values below $10^{-2}$. Therefore, the asymptotic expression offers a low-complexity yet precise tool for system design and analysis at target SER. Comparing across different modulation orders, the results clearly quantify the superior energy efficiency of the PDNN-SSK system as $M$ increases. For instance, to achieve an SER of $10^{-3}$, the required $E_b/N_0$ for $M=4$ is approximately 8.5 dB, whereas for $M=64$, it decreases to about 5 dB. This demonstrates that raising the transmission efficiency from 2 to 6 bits/channel use achieves a considerable SNR gain. This trend starkly contrasts with conventional M-QAM schemes, where increasing the transmission efficiency would incur an SNR penalty to maintain the same error rate. This result confirms that the proposed architecture allows for an energy-efficient scaling of the modulation order, a characteristic that, combined with the inherent light-speed signal processing of PDNNs, highlights its potential for future high-throughput and energy-efficient communication systems.

\begin{figure}[t]
    \centering
    \includegraphics[width=0.5\textwidth]{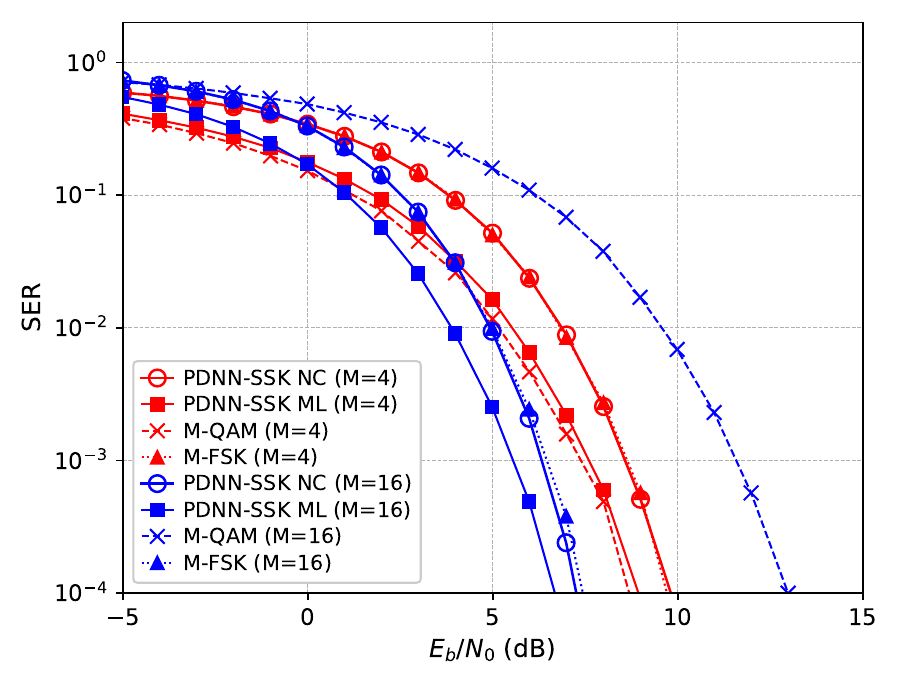}
    \vspace{-0.6cm}
    \caption{SER performance comparison of the PDNN-SSK modulation scheme with conventional ones for $M\in\{4, 16\}$. For the proposed scheme, both a low-complexity NC detector and an optimal ML detector are evaluated.}
    \label{fig:SER_Comparison}
    \vspace{-0.4cm}
\end{figure}

In Fig.~\ref{fig:SER_Comparison}, we further compare the SER performance of the PDNN-SSK system with several benchmark schemes for modulation orders $M\in\{4, 16\}$. Two detectors are considered for the proposed PDNN-SSK: a low-complexity non-coherent (NC) detector that selects the symbol corresponding to the output port with the maximum signal power, and an optimal coherent maximum likelihood (ML) detector. For the ML detector, the estimated symbol index $\hat{m}$ is given by 
\begin{align}
    \hat{m} = \arg\max_{m} \, \mathfrak{R}\{c_{m,m} + n_m\},
\end{align}
where $\mathfrak{R}\{\cdot\}$ denotes the operation of taking the real part. The benchmark schemes include conventional M-ary QAM and non-coherent M-ary FSK. The results show that the performance of the non-coherent PDNN-SSK is nearly identical to that of M-FSK. This is because M-FSK achieves orthogonality between signals in the frequency domain, which is analogous to the spatial orthogonality between the effective channels in the PDNN-SSK under ideal conditions, resulting in a similar mathematical structure for the SER. A trade-off is observed between the two detectors for the PDNN-SSK scheme; the simpler NC detector incurs an SNR penalty of 1-3 dB compared to the optimal ML detector, and this performance gap gradually narrows as the SNR increases. Furthermore, a notable trend is observed with M-QAM: When $M=4$, it is highly energy-efficient due to its effective use of both in-phase and quadrature dimensions. However, as the modulation order increases to $M=16$, its energy efficiency degrades significantly relative to the other schemes. This occurs because, unlike QAM which crowds more points into the same dimensions, schemes like M-FSK and the proposed PDNN-SSK leverage additional dimensions (frequency or space) for modulation, allowing for more graceful scaling in energy efficiency at the potential cost of lower spectral efficiency or multiplexing gain.

\subsection{The Training Approach}

\begin{figure}[t]
    \centering
    \includegraphics[width=0.48\textwidth]{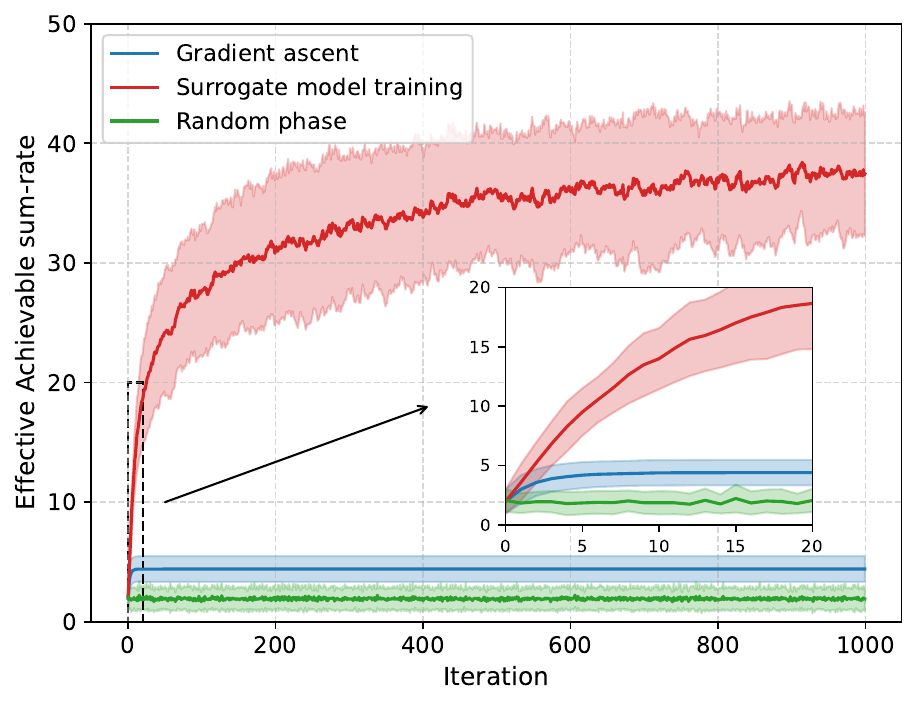}
    \vspace{-0.1cm}
    \caption{Performance comparison of different methods. The solid line represents the mean value over all channel realizations, while the shaded area indicates the standard deviation.}
    \label{fig:convergence}
    \vspace{-0.4cm}
\end{figure}

Fig.~\ref{fig:convergence} illustrates the convergence behavior of the surrogate model training algorithm and compares it with two benchmark schemes: a conventional projected gradient ascent method and a random phase configuration. For this analysis, the system is configured with a modulation order of $M=4$, and the PDNNs at both the transmitter and receiver feature $L_\text{T}=L_\text{R}=2$ layers with $N=16$ ports each. To ensure statistical robustness against the randomness of the channel, the performance is averaged over 100 independent channel realizations. The surrogate model training utilizes the Adam optimizer with a learning rate of 0.1 for 1000 epochs, while the projected gradient ascent is equipped with Armijo step~\cite{anStackedIntelligentMetasurfaces2025}, which leverages the backtracking line search at each iteration. The results demonstrate the superior performance of the surrogate model training approach, which rapidly and consistently improves the effective achievable sum-rate and exhibits stable convergence. In contrast, the projected gradient ascent method shows negligible improvement after the initial several iterations, quickly becoming trapped in a local optimum, as highlighted by the inset figure. The random phase approach, serving as a worst-case baseline, yields an extremely low sum-rate, indicating a lack of effective beamforming. This comparison underscores the ability of the training method to effectively navigate the high-dimensional and non-convex optimization landscape, finding a high-quality solution where traditional optimization methods fail.

\begin{figure}[t]
    \centering
    \includegraphics[width=0.48\textwidth]{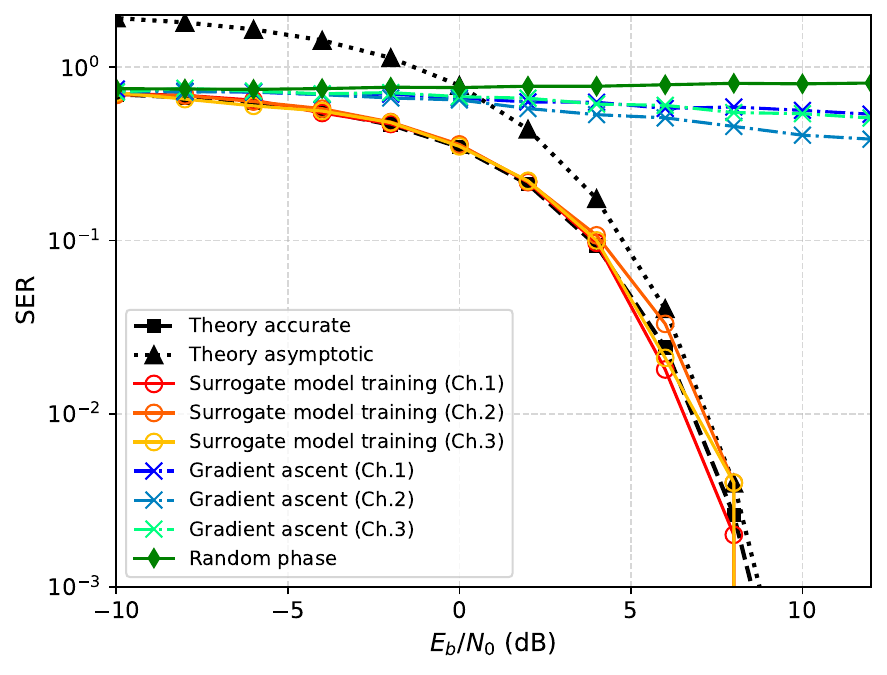}
    \vspace{-0.1cm}
    \caption{Simulated SER performance versus SNR with optimized PDNNs. The theoretical curves are shown for reference.}
    \label{fig:trainedSER}
\end{figure}

The impact of superior optimization performance on the end-to-end communication quality is further validated in Fig.~\ref{fig:trainedSER}. This analysis illustrates the SER performance with the optimized PDNN for 3 distinct channel realizations, with the same system setup as that in Fig.~\ref{fig:convergence}. The theoretical SER curves are plotted based on Theorem \ref{thm:closed_form_ser} and Proposition \ref{cor:asymptotic_ser}, representing the optimal performance achievable in an ideal interference-free condition. The simulated SER curves for the surrogate model training method are shown to closely track the theoretical lower bound, which demonstrates that our algorithm successfully configures the PDNNs to almost completely nullify inter-symbol interference, thereby achieving near-optimal performance. Conversely, the solutions found by the gradient ascent method result in a high SER error floor. In the high-SNR regime, where noise effect is negligible, the performance of this method is limited by the strong residual inter-symbol interference it failed to eliminate. This causes the relevant SER curves to saturate at a constant value, rather than continuing to decrease. The random phase scheme keeps performing the worst, and at higher SNRs, the presence of interference even causes the SER to increase. These results confirm that the high achievable sum-rate attained by our training approach translates into outstanding SER performance, validating the effectiveness of the sum-rate objective and the proposed method for training the PDNN-SSK system.

\subsection{Effects of Structural Parameters}

\begin{figure}[t]
    \centering
    \includegraphics[width=0.48\textwidth]{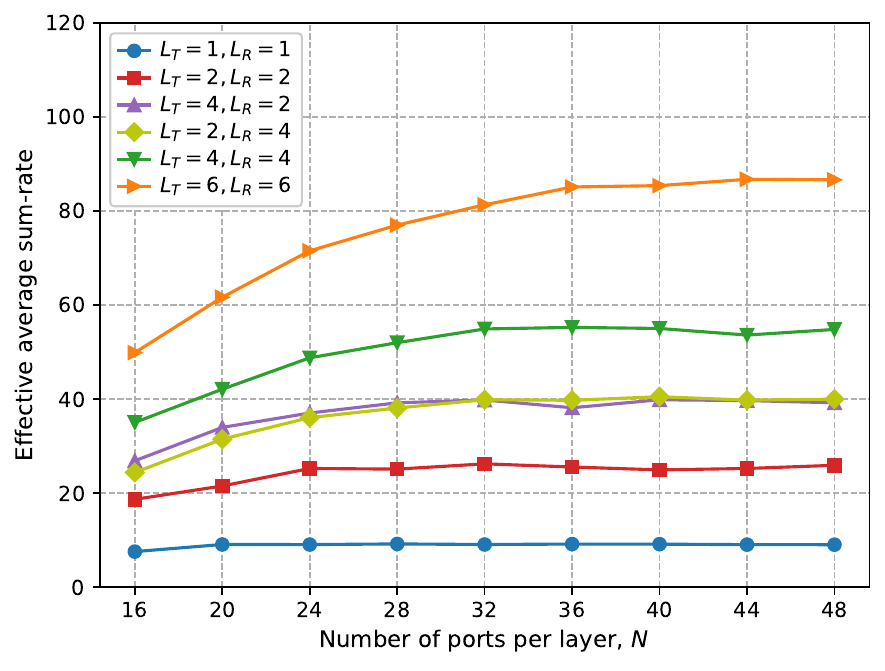}
    \vspace{-0.1cm}
    \caption{Effective achievable sum-rate as a function of $N$ for various PDNN layer configurations.}
    \label{fig:rate_vs_NandLTLR}
    \vspace{-0.4cm}
\end{figure}

Fig.~\ref{fig:rate_vs_NandLTLR} illustrates the effective average sum-rate as a function of the number of ports per layer ($N$) for various PDNN layer configurations ($L_\text{T}, L_\text{R}$). In this simulation, we set $M = 16$, and the results are averaged over 100 independent channel realizations. The results indicate that the depth of the PDNN is a dominant performance factor, with a significant sum-rate enhancement observed as the total number of layers increases, which is attributed to the enhanced signal processing capability of deeper networks. The nearly identical sum-rates for asymmetric configurations suggest that PDNNs can be deployed with spatial asymmetry to suit specific hardware constraints without a significant performance penalty. Moreover, the impact of the number of ports is strongly correlated with network depth. For the shallowest network ($L_\text{T} = L_\text{R} = 1$), the sum-rate shows no improvement with increasing $N$, as spatially distant ports in a shallow network cannot be effectively coupled. For deeper configurations, the sum-rate generally improves as $N$ increases but with diminishing returns, and the saturation point shifts to higher values of $N$ as the depth increases. For example, the curve for $L_\text{T} = L_\text{R} = 2$ saturates around $N = 24$, but the curve for $L_\text{T} = L_\text{R} = 6$ continues improving until $N = 36$, suggesting that the increased depth allows more effective utilization of a wider network.

\begin{figure}[t]
    \centering
    \includegraphics[width=0.48\textwidth]{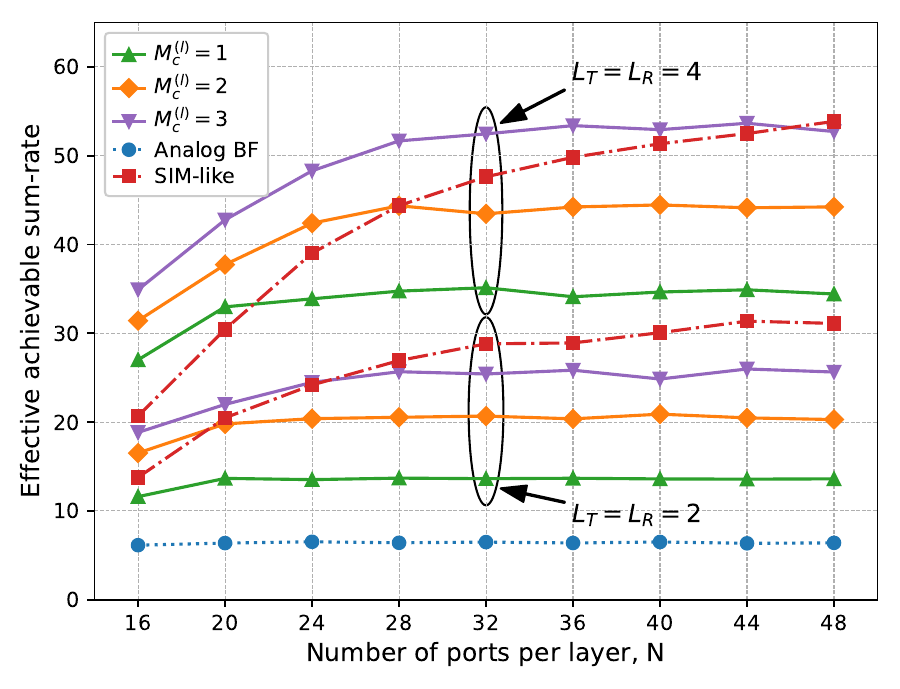}
    \vspace{-0.1cm}
    \caption{Effective achievable sum-rate comparison for different PDNN coupling configurations $(L_\text{T}, L_\text{R})$.}
    \label{fig:rate_vs_NandMl}
    \vspace{-0.4cm}
\end{figure}

Further investigation into the physical interconnection structure is presented in Fig.~\ref{fig:rate_vs_NandMl}, which compares coupler-based PDNNs, a SIM-like planar diffraction network \cite{wangIntegratedPhotonicMetasystem2022}, and a conventional analog beamforming baseline, under simulation conditions consistent with the previous analysis. The SIM-like network models free-space diffraction, with its interconnection matrix $\mathbf{W}^{(l)}$ constructed using the Rayleigh-Sommerfeld formula:
\begin{align}
    w_{m, \tilde{m}}^{(l)} = \frac{A_l \cos \chi_{m, \tilde{m}}^l}{r_{m, \tilde{m}}^l} \left( \frac{1}{2\pi r_{m, \tilde{m}}^l} - j \frac{1}{\lambda} \right) e^{j2\pi r_{m, \tilde{m}}^l / \lambda},
\end{align}
where $A_l$ is the area of each meta-atom, $\lambda$ is the carrier wavelength (corresponding to $f_c = 28$ GHz), $r_{m, \tilde{m}}^l$ is the distance between ports of adjacent layers, and $\chi_{m, \tilde{m}}^l$ is the angle of propagation. The inter-layer and inter-element spacings are set to $2\lambda$ and $0.5\lambda$, respectively. The results show that all PDNN architectures substantially outperform the analog beamforming baseline, which is equivalent to an $L=1$ system with no coupling or diffraction. For the coupler-based PDNNs, the sum-rate significantly increases with the number of cascaded coupler layers, $M_\text{c}^{(l)}$. For the 4-layer configuration, increasing $M_\text{c}^{(l)}$ from 1 to 3 raises the peak sum-rate by over 50\% for $N\ge 28$. Meanwhile, increasing $M_\text{c}^{(l)}$ enables the network to effectively leverage a greater width, comparable to increasing the layers. In contrast, the SIM-like network does not need artificially designed couplers and is capable of establishing full interconnection between adjacent layers through diffraction. In the 2-layer case, this superior connectivity allows it to outperform all types of coupler-based PDNNs under most of configurations. However, in the deeper 4-layer case, the performance of the SIM-like network is not significantly superior to its coupler-based counterpart and is outperformed by the $M_\text{c}^{(l)}=3$ network in most cases. This is attributed to the fact that coupler-based PDNNs confine signals within low-loss transmission-line waveguides, whereas the free-space diffraction in the SIM-like architecture incurs inherent propagation losses that accumulate over multiple layers.

\begin{remark}
	Our results reveal a fundamental trade-off among network depth ($L$), width ($N$), and in-layer coupling strength ($M_\text{c}^{(l)}$). A network's ability to effectively leverage the increased width is contingent upon sufficient depth and strong internal coupling. As evidenced by the simulations, shallow or weakly coupled networks reach performance saturation quickly with increasing $N$. Conversely, deeper and more strongly coupled architectures continue to scale with $N$, demonstrating that sufficient depth and coupling are prerequisites for leveraging the benefits of a wider network.

\end{remark}

\section{Conclusion}
In this paper, we have proposed a PDNN-empowered communication system that embeds signal processing into the propagation of EM waves through artificially designed planar circuits. By introducing and modeling the PDNN, our work provides a low-latency and energy-efficient signal processing platform. The proposed PDNN-SSK architecture deploys PDNNs to jointly perform modulation, beamforming, and detection of SSK signals directly in the EM domain, thus enabling significantly simplified hardware with a single RF chain and non-coherent power detection. We established a rigorous theoretical foundation for this system, deriving the accurate and asymptotic closed-form SER expressions under the interference-free assumption, which serve as crucial benchmarks for the system's performance. To achieve this condition, we formulated the phase shift configuration problem of PDNNs as an effective sum-rate maximization problem with a non-convex structure, and solved it by developing a surrogate model trained with a gradient-based approach. 

According to the simulation results, our training-based approach was shown to closely track the theoretical performance bounds where conventional gradient ascent methods fail and suffer from a high SER error floor. The results also highlighted the superior energy efficiency and scalability of the modulation scheme against existing technologies like QAM by utilizing spatial resources. Fundamental design principles were revealed; for instance, a deeper structure and stronger internal coupling are essential for PDNNs to effectively widen the network. Looking forward, the PDNN-SSK communication system can be extended to multi-user scenarios, more complex channel conditions, and broadband systems. Moreover, the PDNN architecture can be enhanced with more sophisticated designs to implement advanced functionalities, serving as a versatile RF signal processing front-end for a wide array of future communication and sensing applications.

\bibliographystyle{IEEEtran}
\bibliography{reference.bib}

\end{document}